\documentclass[a4paper,cleveref,autoref,thm-restate,USenglish]{lipics-v2021}

\usepackage[utf8]{inputenc} 
\usepackage{xcolor}
\usepackage{amsmath,amsfonts,amssymb} 
\usepackage{multirow,array}
\usepackage{comment}
\usepackage{enumitem}   

\usepackage{graphicx}
\graphicspath{{./img/}}

\newcommand{\gammab}{\Gamma_1}
\newcommand{\gammao}{\Gamma_2}

\newcommand{\gammano}{\Gamma_3}

\newcommand{\pib}{\Pi_1}
\newcommand{\pio}{\Pi_2}

\newcommand{\pino}{\Pi_3}

\theoremstyle{definition}

\title{Foundations of Fiat-Denominated Loans Collateralized by Cryptocurrencies}
\titlerunning{Foundations of Fiat-Denominated Loans Collateralized by Cryptocurrencies}

\author{Pavel Hubáček}
{Institute of Mathematics, Czech Academy of Sciences, Prague, Czech Republic
\and
Charles University, Faculty of Mathematics and Physics, Prague, Czech Republic}
{hubacek@math.cas.cz}
{https://orcid.org/0000-0002-6850-6222}
{Partially supported by the Academy of Sciences of the Czech Republic (RVO 67985840) and Czech Science Foundation GAČR grant No. 25-16311S.}
\author{Jan Václavek}{Firefish, Prague, Czech Republic}{}{}{}
\author{Michelle Yeo}{National University of Singapore, Singapore}{michellexyeo@gmail.com}{https://orcid.org/0009-0001-3676-4809}{Supported by MOE-T2EP20122-0014 (Data-Driven Distributed Algorithms)}

\authorrunning{P. Hubáček, J. Václavek, and M. Yeo}

\Copyright{Pavel Hubáček, Jan Václavek, and Michelle Yeo}

\hideLIPIcs

\nolinenumbers

\ccsdesc[500]{Applied computing~Electronic commerce}
\ccsdesc[300]{Applied computing~Economics}

\keywords{Blockchains, Cryptocurrencies, DeFi, Loans, Mechanism design, Subgame Perfect Equilibrium, Rational analysis}

\begin{document}
\maketitle

\begin{abstract}
The rising importance of cryptocurrencies as financial assets pushed their applicability from an object of speculation closer to standard financial instruments such as loans.
In this work, we initiate the study of secure protocols that enable fiat-denominated loans collateralized by cryptocurrencies such as Bitcoin.
We provide limited-custodial protocols for such loans relying only on trusted arbitration and provide their game-theoretical analysis.
We also highlight various interesting directions for future research.
\end{abstract}

\section{Introduction}

The increase in the perceived value of cryptocurrencies over the last decade gave rise to a multitude of financial products and, in fact, a whole domain known as \emph{Decentralized Finance} (DeFi)~\cite{WernerEtAlSoKDeFi,schar2021decentralized}.
In this work, we study loans collateralized by cryptocurrencies.
At the moment, there are various providers of such loans (e.g., \href{https://aave.com/}{Aave}, \href{https://compound.finance/}{Compound}, \href{https://debifi.com/}{DebiFi}, or \href{https://lend.hodlhodl.com/}{HodlHodl}).
Some estimate that over 400,000 BTC could have been used as collateral in the cryptocurrency lending market in 2021~\cite{BTCCollateral}.

On a high level, any such lending protocol contains \emph{Lenders} bringing liquidity, \emph{Borrowers} wishing to use cryptoassets as collateral, and a \emph{Platform} connecting the Borrowers with Lenders and handling the loans.

\smallskip{\em Current approaches to loans collateralized by cryptocurrencies.}
Most current approaches to cryptocurrency-backed loans today are centralized \emph{custody-based solutions} where the custodian has full control over the collateral for the lifetime of the loan.
There are two most common types, with majority in the second category: 

\begin{description}
\item[Peer-to-peer:] where the custodian plays a role of a platform that sets up and manages loans (e.g. \href{https://lendabit.com/}{LendaBit.com}), and 
\item[Bank-style:] where the custodian is both the Lender and Platform (e.g. \href{https://blockfi.com/crypto-loans/}{BlockFi}).
\end{description}

In both formats, the custody-based approach has its undeniable advantages: ease of loan setup, minimal interaction between the parties, and straightforward liquidation or top-up of collateral.
On the other hand, the clear major disadvantage, which outweighs all the advantages for many borrowers, is that the custodian has full control over the collateral.
This can lead to disastrous scenarios:
\begin{itemize}
\item The custodian may provide collateral to another entity for further profit (a practice commonly referred to as \emph{rehypothecation}); this exposes collateral to significant counterparty risk as seen in the crash of cryptoloans markets in July 2022~\cite{CryptoloansCrash,GudgeonPHLG20}.
\item The custodian can steal the collateral at any time and there is no mechanism to prevent this.
\item As a single point of failure, the custodian can be hacked, be compromised, or go out of business, and all (or part of) the collateral gets lost.
\end{itemize}

Note that the collateral above refers to the combined collateral from all the loans managed by a custodian at any given time, which can add up to a vast amount of bitcoin, making it a high-value target for the custodian, its employees, hackers, etc.

The second approach, which is arguably not as common, are \emph{multisignature-based solutions} leveraging 2-of-3 multisignature escrow for the collateral.
As with the custody-based approach, there are two types:

\begin{description}
\item[Peer-to-peer:] where the Lender, Borrower and Platform each have a single key to the 2-of-3 multisignature address. The Platform plays a role of an arbitrator and resolves disputes (e.g. \href{https://lend.hodlhodl.com/}{HodlHodl}).    
\item[Bank-style:] where Platform is also the Lender. The Borrower and Platform each have a single key and the third key is held by a so called \emph{third-party key agent}.
The third-party key agent is used when the Borrower does not cooperate or if the Platform goes out of business (e.g. \href{https://unchained.com/loans/}{Unchained Capital}).
\end{description}

The multisignature-based approach is definitely an improvement in terms of security compared to the custodian approach as it eliminates a single point of failure and the complete counterparty risk but it still has significant drawbacks:
In the peer-to-peer format, the Lender must have a private key to the multisignature address and be able to sign specific transactions based on an event that occurs (in the bank-style model, the Lender is essentially the Platform).
Such a requirement may deter many potential Lenders, particularly institutions that may not have the competence or desire to securely manage private keys.
For both formats, the Borrower must have a long-term (\emph{static}) private key to the multisignature address.
The consequence is that the Borrower must securely store such a private key and be ready and able to sign transactions when an event occurs.

\subsection{Fiat-Denominated Loans Collateralized by Cryptocurrencies}
In this work, we focus on a mixed setting where parties interact both with a classical fiat currency such as USD and a cryptocurrency such as Bitcoin.
This differs from the standard approach taken in DeFi, where the protocols are designed primarily  for interactions among digital assets only.
To the best of our knowledge, such a setting was not formally studied before in the context of cryptoloans.
Nevertheless, there is a clear motivation for design of secure protocols supporting such loans.

On the borrower's side, a known speculation strategy is the so-called ``hodling,'' i.e., the strategy to buy some amount of cryptocurrency and
wait an extremely long time for it to rise in value.
Such speculators, known as  hodlers, might wish to use their cryptocurrency as collateral for a loan (e.g., instead of a classical mortgage) as long as they keep  exclusive access to their crypto asset during the lifetime of the loan (except in the case of their default).
An estimated 10M Bitcoins were in dormant wallets in 2020~\cite{DormantBTC}.
On the lender's side, many classical providers of liquidity would be attracted by the interest from an algorithmically governed loan, which is extremely low-risk since the collateral can be liquidated fast in the case of default or undercollateralization due to a drop in value of the cryptocurrency.

However, a purely DeFi solution where the loan is denominated in a cryptocurrency might not be acceptable to either party.
The most natural solution in this context is to use a stablecoin, a decentralised analogue of a fiat currency pegged to its value~\cite{ClarkDM20}, as an intermediary asset for enabling the interface between fiat currency and cryptocurrency.
The first drawback is that stablecoins are built using complex protocols which might deter conservative lenders and borrowers such as a classical bank. 
Additionally, both the lender and borrower (in case the borrower seeks liquidity and not just an instrument for speculation) then have to go through the additional process of exchanging fiat currency for a stablecoin (and vice versa) which induces extra inefficiencies to the process in terms of fees and delays. 

\subsection{Our Protocols}
We present three protocols.
The first protocol illustrates the basic considerations in a simplified setting where we assume exchange rates between the underlying cryptocurrency and fiat are fixed.
Next, we give two protocols in a realistic setting where exchange rates can fluctuate.
All our protocols involve a rational borrower, a rational lender, and an honest arbiter. 
Naturally, the separation of roles in our protocols implies that our protocols have add more value in the more challenging peer-to-peer setting rather than the simpler bank-style setting.

The idea behind the basic protocol $\pib$ is to lock the collateral for the entire duration of the loan into a smart contract.
The contract redistributes the collateral at the end of the loan in a way such that both borrower and lender end up with slightly less than what they put into the protocol, with the remaining portion being held by the arbiter.
The arbiter gives back the respective portions to the parties that behave honestly, and keeps the corresponding portion in case of misbehaviour.

However, in the setting where exchange rates fluctuate thoughout the loan term, the value of the collateral relative to the principal can change drastically and collateral can become vastly overcollateralised or undercollateralised.
This can lead to misbehaviour on the side of both parties.
For instance, the borrower might run away with the loan if the collateral becomes undercollateralised.
Our first protocol $\pio$ in this setting allows the smart contract to make calls to an external price oracle at specified times during the loan term to get the price of the cryptocurrency.
The collateral is then redistributed according to this price to ensure that the lender should always get back an amount worth the value of the principal when the lender behaves honestly. 
When the value of the collateral reaches or drops below the principal, the contract immediately liquidates the collateral and gives the entire collateral to the lender. 

Our second protocol $\pino$ in this setting simplifies the contract significantly by removing oracle calls, and also giving both parties the additional freedom to decide when and at what price to terminate or liquidate the protocol.
The collateral is immediately released upon termination or liquidation. 
This comes with the cost of both parties needing to remain online throughout the loan duration, and the way the collateral is redistributed upon termination or liquidation of the protocol is fixed and has to be specified at the beginning of the protocol.

We show in~\Cref{thm:basic} and~\Cref{thm:main_oracle1} that behaving honestly as specified by the protocol is a subgame perfect equilibrium in the games induced by $\pib$ and $\pio$. 
We also show in~\Cref{thm:main_oracleless} that as long as the price of the cryptocurrency does not rise too much, the honest strategy profile of behaving according to the protocol is a subgame perfect equilibrium in the game induced by $\pino$.

\subsection{Related Work}
The study of the interplay between CeFi and DeFi goes back to Danezis and Meiklejohn~\cite{DanezisM16}.
Qin et al.~\cite{QinZGJG21} conducted an empirical study of major liquidation protocols and their risks in Ethereum.
Gudgeon et al.~\cite{GudgeonW0K20} analysed the effect of interest rates on market efficiency for DeFi loan protocols.
Qin et al.~\cite{QinZLG21} studied attacks on DeFi using flash loans.
Kondratiuk et al.~\cite{KondratiukSNT21} initiated a standardization attempt for cryptoloans using smart contracts on Cardano.

Our work is also related to the works of Avarikioti et al.~\cite{AvarikiotiLW20,AvarikiotiKWZ21} on using a committee of miners for arbitration in payment channels\cite{decker2015fast,csur21crypto,raiden,poon2015lightning} when two rational parties disagree on the state of the channel.
While using a committee or even the blockchain to resolve disputes between the lender and borrower is a possibility, these solutions typically incur a large consensus cost and are less efficient, thus we choose to use the arbiter to resolve disputes. 

\subsection{Glossary}\label{sec:Glossary}
Here, we define terms related to loans that we use throughout the paper. 

\begin{description}
\item[Collateral:] The amount provided by the borrower to back up the loan. The collateral is given to the lender in case the borrower defaults on the loan.
\item[Principal:] The loan amount borrowed by the borrower.
\item[Liquidation:] The process of using the collateral for repaying the loan in the case of default or undercollateralization.
\item[Loan to Value Ratio (LTV):] The ratio of the principal to the value of the collateral. The LTV determines the maximum amount that the borrower can borrow depending on the amount of collateral the borrower has. For instance, an LTV of $\frac{1}{2}$ would mean that a borrower with $2x$ amount of collateral can only borrow a maximum of $x$ amount of funds from the lender (assuming a 1:1 exchange rate). In our work, all our protocols assume an initial LTV of $\frac{1}{2}$.
\item[$\epsilon$:] A small positive constant parameterising the strength of the penalty that comes with deviating from the specified protocol.
\end{description}

\section{Model}\label{sec:model}

\smallskip{\em Notation.} We use BTC or Bitcoin to denote the underlying cryptocurrency used as collateral, however, we stress that our protocol is not only Bitcoin compatible (see~\Cref{sec:extensions}), but can also be used with other cryptocurrencies. 
We use a fiat to BTC exchange rate of $r$ to denote that $1$M fiat is equal to $r$ BTC. 
Equivalently, the price of BTC is the value of $1$ unit of BTC when converted to fiat.
Thus, a rate of $r$ would imply the price of $1$ BTC is $\frac{1}{r}$M fiat.
For a positive integer $k$, we use $[k]$ to denote $\{1, \dots, k\}$.

\smallskip{\em Loan setting and assumptions.}
Wlog, we assume the borrower wants to borrow $1$M in fiat and holds an amount of Bitcoin which is worth $2$M when converted to fiat, henceforth called the collateral. 
For simplification, we ignore interest rates.
Furthermore, we assume that maturity of the loan is a year.

\smallskip{\em Threat model and availability of arbitration.}
We assume the \emph{borrower} and \emph{lender} are rational and have access to an honest party called \emph{arbiter}.
Recall that we consider a mixed setting where the parties have to exchange the fiat currency and lock the collateral into a smart contract.
This situation gives rise to the so-called \emph{fair exchange problem} known also from online commerce, which might not be resolvable using smart contracts without trusted third party~\cite{Goharshady21}.
Given that we envision a system where a Platform sets up the loans (e.g. enables matching of borrowers and lenders), we assume that the Platform can participate as a trusted intermediary and solve the fair exchange of fiat.
Note that, from the perspective of the lifetime of the loan, the participants need to trust the Platform only at the very beginning and very end of the loan (i.e., during the transfer of the fiat from the lender to the borrower and while repaying the loan).
Thus, we allow the arbiter to handle the transfer of fiat and concentrate on the problem of governance of collateral during the lifetime of the loan.
In particular, one of our core design goals is to minimize the trust put into the arbiter when resolving disputes over the collateral.

\smallskip{\em Game theoretic notions and solution concept.}
Because our loan protocols unfold over time with ordered moves and observable intermediate events, we model them as \emph{extensive-form} games.
This representation records who moves when and what each party knows at that moment, enabling reasoning about every feasible history of moves.
Our solution concept for analyzing the behaviour of rational parties is \emph{subgame-perfect equilibrium} (SPE).
A strategy profile is an SPE if, in every subgame, i.e., after every history of moves, no player can profit from a unilateral deviation.
For our protocols, proving SPE means showing by backward induction that at each node the prescribed action maximizes the mover’s payoff given the continuation strategy, including nodes reached only after deviations.
This rules out equilibria sustained by non-credible off-path threats and yields robustness to timing perturbations and unilateral deviations.
Formal definitions of extensive-form games and SPE appear in~\Cref{app:gtconcepts}.


\section{Basic Protocol for Flat Exchange Rates}
In this section, we present a basic protocol $\pib$ for the setting where the exchange rate between Bitcoin and fiat is fixed throughout the entire duration of the loan.
For ease of presentation, we assume a rate of $1$, and we stress that our protocol can be easily tweaked to accommodate other exchange rates.

\subsection{Protocol Details}\label{sec:prot_details}
Our protocol $\pib$ consists of $2$ phases: contract creation and loan repayment. 
In the contract creation phase, the borrower locks the loan collateral into a smart contract and specifies the conditions upon which the contract can be opened.
The loan repayment phase occurs at the end of the loan and consists of extracting the collateral from the contract and distributing it to the relevant parties.
The details of both phases are as follows:

\begin{figure}[htb!]
    \centering
    \includegraphics[width=\textwidth]{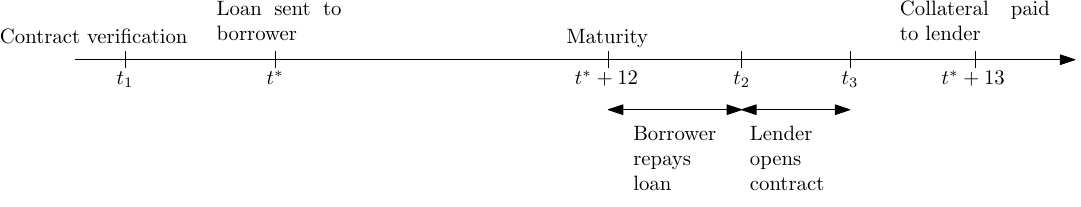}
    \caption{Timeline of important events in the basic protocol.}
    \label{fig:protocol_timeline}
\end{figure}

\begin{description}
\item[Contract creation:]\hfill
\begin{itemize}
    \item The lender sends $1$M fiat to the arbiter.
    \item The borrower creates a smart contract and locks the collateral into the contract. 
    The borrower also specifies three important times: $t^*, t_2, t_3$ with $t_3 > t_2 > t^* + 12$ and $t^*$ being the start of the loan term. 
    The smart contract only accepts inputs from the arbiter from time $t^* +12$ to $t_2$, inputs from the lender from time $t_2$ to $t_3$, and inputs from the borrower from time $t_3$ to $t^*+13$ (more details in the loan repayment phase below).
    \item The collateral is locked until $1$ out of the following conditions are fulfilled:
    \begin{enumerate}
        \item At the end of the loan term the lender inputs a signed opening to the contract.
        The contract then pays the full collateral to the borrower.
        \item At the end of the loan term, the borrower inputs a signed opening to the contract. 
        The contract then pays out Bitcoins worth $(1-\epsilon)$M fiat to the borrower, Bitcoins worth $(1-\epsilon)$M fiat to the lender, and Bitcoins worth $2\epsilon$M fiat to the arbiter.
        \item At one month after the maturity date (i.e. $t^*+13$), the contract automatically sends the full collateral to the lender.
    \end{enumerate}
    \item The arbiter verifies the contract.
    If the contract does not pass the verification step, the arbiter refunds the principal to the lender. 
    If the borrower does not create the contract before a timeout period $t_1$, the arbiter refunds the principal to the lender. Otherwise, the arbiter sends the principal to the borrower at time $t^*$.
\end{itemize}

\item[Loan repayment:]\hfill
\begin{itemize}
    \item From time $t^* +12$ to $t_2$, the smart contract waits for the borrower to send $1$M to the arbiter. 
    If the borrower repays the debt within this time interval, the arbiter immediately notifies the lender.
    \item The lender waits until the arbiter notifies the lender that they have the full principal and inputs a signed opening to the contract at any point between $t_2$ and $t_3$ to send the collateral to the borrower.
    After the contract is opened by the lender within the specified time interval, the arbiter sends $1$M to the lender.
    \item From time $t^*+12$ to $t^* + 13$, two mutually exclusive events can happen: 
    \begin{itemize}
        \item (Honest borrower, dishonest lender.) The borrower sent the full principal but the lender does not open the contract between $t_2$ and $t_3$, so the borrower triggers an opening of the contract after $t_3$ and before $t^*+13$ months.
        In this case, the arbiter gives the principal to the borrower and $\epsilon$M out of the $2\epsilon$M worth of Bitcoins paid to the arbiter from the contract to the borrower.
        The borrower ends up with $2$M ($(1-\epsilon)$M from the contract, $1$M from the loan, and $\epsilon$M from the arbiter). 
        The lender ends up with $(1-\epsilon$)M from the contract, and the arbiter ends up with $\epsilon$M.
        \item (Dishonest borrower, honest lender.) The borrower did not send the full principal and triggers an opening of the contract after $t_3$.
        Suppose the partial principal sent to the arbiter by the borrower is $x$M where $x<1$. 
        The arbiter keeps the partial principal of $x$M and sends $\epsilon$M from their share of the contract to the lender.
        The lender ends up with $1$M ($(1-\epsilon)$M from the contract and $\epsilon$M from the arbiter), the arbiter with $(x+\epsilon)$M, and the borrower with $(1-x+1-\epsilon)$M $<2$M.
    \end{itemize}
    \item If the contract has not been opened by $1$ month after the loan maturity date (time $t^*+13$), the contract automatically pays out the locked collateral to the lender.
\end{itemize}
\end{description}

\subsection{Protocol Security}
\smallskip{\em Game tree and honest strategy.}
We first note that $\pib$ induces a $5$ stage extensive-form game $\gammab$ which we depict in~\Cref{fig:basic_game_tree}.
The game $\gammab$ is played among 2 players, the borrower and the lender.
In the first stage of $\gammab$, only the lender can make a move. The actions available to the lender at this stage can be represented by $\{lend, \neg lend\}$. The action $lend$ represents the lender lending the full loan amount, while $\neg lend$ denotes the set of all other actions that the lender can play which is not $lend$, including lending only a partial amount. Only when the lender plays $lend$ will the game proceed to stage 2. Otherwise, the game terminates at this stage with both lender and borrower having a utility of $0$.

Stage 2 of $\gammab$ is the start of the contract creation phase where the borrower creates a smart contract as specified in the contract creation phase in~\Cref{sec:prot_details}.
Only the borrower can make a move at this stage and the actions available to the borrower can be represented by $\{correct, incorrect\}$. The action $correct$ represents the borrower creates a correct contract (as verified by the arbiter), and $incorrect$ represents the set of all other actions that are not $correct$. Only when the borrower plays $correct$ will the game proceed to stage 3. Otherwise, the game terminates at this stage with both lender and borrower having a utility of $0$.
Stage 2 of $\gammab$ starts at time $t_1$ as depicted in~\Cref{fig:protocol_timeline}. 

Successful contract creation and verification results in the full principal sent by the lender to the borrower, corresponding to time $t^*$ in~\Cref{fig:protocol_timeline}. 
After the loan matures, which corresponds to time $t^*+12$ in~\Cref{fig:protocol_timeline}, the loan repayment phase begins, which corresponds to stage 3 of $\gammab$.
The loan repayment phase spans a duration from time $t^*+12$ to $t_2$ as depicted in~\Cref{fig:protocol_timeline}.
In stage 3, only the borrower makes a move. The actions available to the borrower is the choice of $\{\bot, x\}$ for $x \in [0,1]$M. The action $x$ denotes the amount the borrower decides to repay, while $\bot$ denotes any other action which is not loan repayment.  Wlog, we model playing $\bot$ as the same as playing $x=0$ (i.e., the utilities at the leaves of the game tree depicted in~\Cref{fig:basic_game_tree} corresponding to playing $\bot$ is the same as playing $x=0$).
After the borrower plays their move, the game proceeds to stage 4 with the move of the lender.

In stage 4 of $\gammab$, the lender responds to the borrower's actions in stage 3 with a choice to open the contract and release the collateral or not. 
This stage spans a duration from time $t_2$ to $t_3$ as depicted in~\Cref{fig:protocol_timeline}.
In this stage, only the lender makes a move. Formally, the actions available to the lender can be represented by $\{open, \neg open\}$. The action $open$ represents the lender opening the contract and releasing the collateral, while $\neg open$ represents all other actions that are not $open$.
After the lender makes their move, the game proceeds to the final stage. 

Since the contract can only be opened once, stage 5 of $\gammab$ corresponds to the choice of the borrower to open the contract if the lender does not open the contract at the previous stage. 
Only the borrower makes a move at this stage and this stage happens after time $t_3$ and before time $t^*+13$ as depicted in the protocol timeline in~\Cref{fig:protocol_timeline}.
The actions available to the borrower can be represented by $\{open, \neg open\}$. The action $open$ represents the borrower opening the contract, while $\neg open$ represents all other actions that are not $open$. 

\begin{figure}[htb!]
    \centering
    \includegraphics[width=\textwidth]{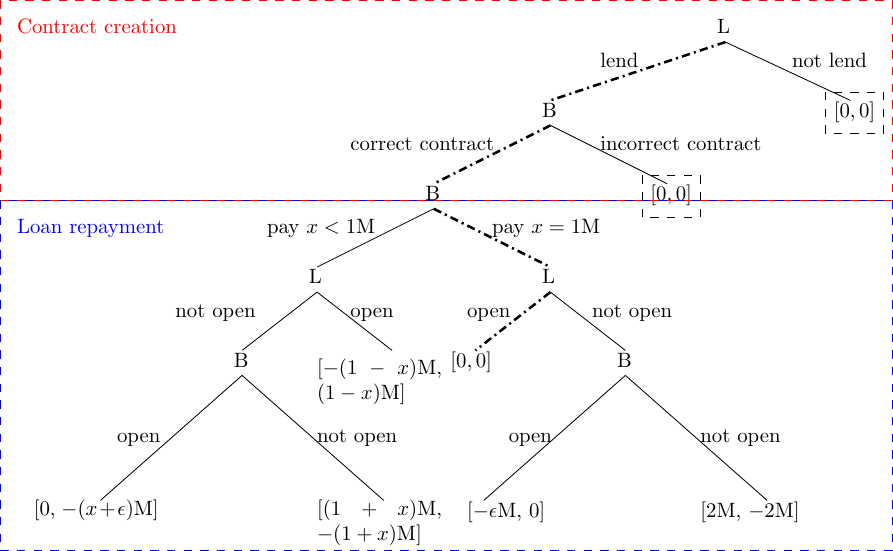}
    \caption{Game tree induced by $\pib$ showing the actions the borrower and lender can take at each step of the protocol. 
    Vertices labelled $L$ are lender vertices and vertices labelled $B$ are borrower vertices. 
    The first element of the utility vector (in fiat) at each leaf node corresponds to the utility of the lender, and the second element the utility of the borrower. 
    The dotted and dashed path depicts the honest strategy.}
    \label{fig:basic_game_tree}
\end{figure}

Let us define the honest strategy profile as $\sigma =(\text{lend}, \text{correct contract}, x=1\text{M}, \text{open})$, i.e., the dotted and dashed path in \Cref{fig:basic_game_tree}. 
The following theorem (with proof in~\Cref{app:proof_basic}) shows that $\sigma$ is a subgame perfect equilibrium in $\gammab$

\begin{theorem}\label{thm:basic}
$\sigma$ is a subgame perfect equilibrium in $\gammab$.
\end{theorem}

\begin{remark}\label{rem:delta}
Note that since the utility vectors at stages $1$ and $2$ of the reduced subgame are exactly the same, we have multiple equilibria, some of which correspond to strategies that place a non-zero probability on either not lending the loan or not creating a correct contract.
It is reasonable to eliminate such strategies, as they are not meaningful to our analysis of the security of the protocol.
As such, we assume the borrower and lender prefer that a deal occurs (i.e., the borrower successfully gains the loan) than not, and thus we replace the utility vectors of the cases corresponding to a failed deal (i.e., the positions corresponding to the dashed rectangles in \Cref{fig:basic_game_tree}) with $(-\delta, -\delta)$ for some $\delta > 0$. 
\end{remark}

Finally, we also outline some minor extensions of $\pib$ to shield the lender against an unresponsive borrower, and to allow both parties to remain offline during loan repayment.

\smallskip{\em Dealing with an unresponsive/slow borrower.}
We outline an extension of $\pib$ to protect the lender in the case where the borrower might be unresponsive or slow.
This is the case where the lender already sent the principal to the arbiter, but the borrower is unresponsive/slow, and has not created the contract locking the collateral. 
We note that the lender will eventually get back their principal, assuming the arbiter waits for the borrower to act within a prespecified timeout period.
However, the lender runs the risk of having their money locked up with the arbiter for potentially the entirety of the timeout period, incurring a corresponding opportunity cost.

This can be mitigated by splitting the principal into smaller chunks, and executing the protocol separately on each chunk of the loan.
Specifically, the parties agree on a parameter $k$ and execute the contract creation phase of the protocol $k$ times, each time dealing with $\frac{1}{k}$M of the principal.
In each iteration of the contract creation phase, the lender sends $\frac{1}{k}$M to the arbiter.
The borrower locks $\frac{2}{k}$M into the smart contract as collateral (held for the same loan term) with the possibility of paying out $\frac{1-\epsilon}{k}$M to the borrower, $\frac{1-\epsilon}{k}$M to the lender and $\frac{2\epsilon}{k}$M to the arbiter at the end of the maturity date if the borrower inputs a signed opening.
Loan repayment can then be done for each contract separately.

\smallskip{\em Replacing lender and borrower with arbiter in loan repayment.}
We note that since the arbiter is honest, we can replace the actions of the borrower and lender in the loan replacement phase with the arbiter, thereby making it possible for the borrower and lender to remain offline once the borrower pays off the loan. 

\section{Protocols for Fluctuating Exchange Rates}
In this section, we present two protocols that account for fluctuating exchange rates during the loan term.
The first protocol $\pio$ uses an external, third party oracle to constantly update the amount of collateral in the contract such that both the borrower and lender will still be incentivised to follow the protocol. 
The second protocol $\pino$ is a no oracle solution that simply allows the lender and the borrower to liquidate the collateral at any time point.
We work in the same loan setting as described in~\Cref{sec:model}.
That is, the lender loans $1$M fiat to the borrower, and the borrower locks $y$ BTC in the contract where $y$ is the amount of BTC worth $2$M in fiat at the time of contract creation.

\subsection{Using a Price Oracle to Modify the Collateral}\label{sec:oracle}
Here, we describe our first protocol $\pio$, which contains a few small modifications of the basic protocol to ensure the incentives of the borrower and lender are still aligned even though the exchange rates can fluctuate. 
The first modification allows the entire collateral to be given to the lender the moment the price of BTC falls below a certain value.
This compensates the lender in the case where the price of BTC drops and ensures that the value of the collateral is always at least as large as the loan. 
Consequently, the borrower and lender are still incentivised to follow the protocol. 
The second modification allows the contract to make queries to an external price oracle to get the price of BTC and redistribute the collateral according to the updated price.
Our protocol $\pio$ aims to maintain the invariant that the value of the collateral at any queried point in time should not exceed $2$M when converted to fiat -- any excess collateral is transferred to a temporary account for use when the value of the collateral drops below $2$M or given back to the borrower in the case where $\pio$ terminates with the excess collateral. 
This ensures that the incentives of the lender and borrower are the same as the basic protocol $\pib$ in the case where BTC prices go up. 
To ensure that the lender and borrower cannot easily collude with the external oracle and manipulate the exchange rate to their advantage, we recommend using a decentralised oracle connected to a blockchain.

\subsubsection{Protocol Details}\label{sec:oracledetails}
Let $r_0$ be the exchange rate at the start of the protocol and $r_t$ be the exchange rate at maturity.
Let $p_0 = \frac{1}{r_0}$ and $p_t = \frac{1}{r_t}$ denote the corresponding price of BTC at the start and maturity.
$\pio$ consists of $3$ phases: contract creation, collateral redistribution, and loan repayment.

\begin{description}
\item[Contract creation:]\hfill
\begin{itemize}
    \item The smart contract has access to an external price oracle that outputs the current price of BTC in fiat, and the borrower (in agreement with the lender) stipulates the frequency of the queries throughout the loan maturity period, with one query at the loan maturity date.
    We use $q\geq 1$ to denote the total number of queries.
    \item The borrower deposits collateral of amount $y$ which should equal $2$M fiat when converted with exchange rate $r_0$ which is the exchange rate at the start of the protocol.
    \item The collateral is locked until $1$ out of the following conditions are fulfilled:
    \begin{enumerate}
        \item (Lender liquidation condition.) If a queried price of BTC $p_i$ drops such that $p_i < \frac{p_0}{2}$, the collateral is immediately released and transferred in full to the lender.
        \item (Borrower early termination condition.) The borrower sets an additional threshold $\tau$.
        If a queried price of BTC $p_i$ rises such that $p_i \geq \tau > p_0$, $\pio$ proceeds immediately as if the loan maturity is now at time $i$ and initiates the loan repayment phase.
        \item (Collateral invariant condition.) If a queried price $p_i$ is above some threshold, a portion of the collateral is transferred to a temporary account owned by the arbiter. See the collateral redistribution phase for more details. 
        \item At the end of the loan term, the lender inputs a signed opening to the contract.
        The contract then pays the full collateral to the borrower.
        \item At the end of the loan term, the borrower inputs a signed opening to the contract. 
        The contract then pays out BTC worth $(1-\epsilon)$M fiat to the borrower, BTC worth $(1-\epsilon)$M fiat to the lender, and BTC worth $2\epsilon$M fiat to the arbiter in the case where the price of the collateral is equal to $2$M in fiat.
        In the case where the price of the collateral amount is $<2$M, the contract pays out BTC worth $(y \cdot p_t-1-\epsilon')$M to the borrower, BTC worth
        $(1-\epsilon')$M to the lender, and BTC worth $2\epsilon'$M to the arbiter for some $\epsilon'>0$ (recall that $y$ is the amount of BTC worth $2$M at contract creation). Note that these sum up to $y\cdot p_t$ which is the value of the collateral in fiat at time $t$.
        \item At one month after maturity (i.e. $t^*+13$), the contract automatically sends the full collateral to the lender.
    \end{enumerate}
\end{itemize}

\begin{remark}
The purpose of the borrower early termination threshold $\tau$ is to provide the borrower with a ``fast-track" to loan repayment.
This is useful in case the value of the collateral is increasing too much and the borrower wants to liquidate it to use the Bitcoins for some other purpose.
As it is difficult to model exactly what this threshold should be (it depends on a myriad of factors, including the borrower's risk propensity, opportunity cost, market conditions, etc.), we do not specify what $\tau$ should be and simply leave $\tau$ as an option open to the borrower. 
\end{remark}

\item[Collateral redistribution:]\hfill

The collateral redistribution phase starts from $t^*$ and ends at the loan maturity $t^*+12$.
\begin{itemize}
    \item At each predetermined query point $i$, the smart contract makes a query to an external price oracle to get the current price BTC $p_i$.
    Recall that $p_0$ is the price of BTC at contract creation time which we assume is $2$.
    \item ($p_i > p_0$.) 
    A significant rise in BTC price benefits the lender as the lender could profit by \emph{not} opening the contract even after the borrower repays the loan, and instead wait for the borrower to open the contract. 
    Doing so when the price of BTC is large enough could ensure that the amount of BTC the lender receives when the borrower opens the contract using condition $3$ is larger than $1$M when converted to fiat, even with the $\epsilon$ penalty.
    To prevent this, our protocol triggers the collateral invariant condition to release a portion of the collateral to the temporary account to maintain the invariant that the price of the updated collateral amount has to always equal $2$M. 
    That is, BTC worth $(y \cdot p_i - 2)$M is transferred to the temporary account.
    \item ($p_i < p_0$.) 
    A significant drop in the price of BTC could harm the lender, as the collateral would now be a lot smaller when converted to fiat.
    Our protocol transfers any excess coins in the temporary account to the collateral while maintaining that the collateral should not exceed $2$M when converted to fiat.
    If the temporary account is empty, the contract does nothing and the liquidation condition might happen.
\end{itemize}

\item[Loan repayment:]\hfill

The loan repayment phase of $\pio$ is similar to the loan repayment phase in $\pib$, with the only changes being the liquidation condition, as well as handling excess funds in the temporary account:
\begin{itemize}
    \item If the liquidation condition of the contract is triggered, the smart contract releases the full collateral to the lender, bypassing loan repayment. 
    \item (Honest borrower, dishonest lender.) The case where $p_t \geq p_0$ follows exactly in the case of the basic protocol. 
    If $p_t < p_0$, the arbiter gives the principal and $\epsilon'$M to the borrower.
    The borrower ends up with $y\cdot p_t$M. 
    The lender ends up with $(1-\epsilon')$M, and the arbiter with $\epsilon'$M.
    \item (Dishonest borrower, honest lender.) 
    The case where $p_t \geq p_0$ follows exactly in the case of the basic protocol. 
    If $p_t < p_0$, suppose the partial principal sent to the arbiter by the borrower is $x$M where $x<1$. 
    The arbiter sends the partial principal of $x$M and $\epsilon'$ BTC from their share of the contract to the lender.
    The lender ends up with $(1 + x)$M ($1-\epsilon'$M paid out from the contract and $(\epsilon'+x)$M from the arbiter), the arbiter with $\epsilon'$M, and the borrower with $(y\cdot p_t -1 - \epsilon')$M.
    Note $0<\epsilon' < y\cdot p_t -1$.
    \item If there are still excess coins in the temporary account at maturity, all excess coins will be paid out to the borrower when the protocol terminates if the borrower gives the full principal to the arbiter.
\end{itemize}
\end{description}

\subsubsection{Protocol Security}
Let $\alpha := \frac{y}{2}p_t - 1$. 
Intuitively, $\alpha$ captures the difference in price of $1$ unit of BTC at the end of the protocol and at the beginning of the protocol assuming an exchange rate of $1$ at the start.
We first make a few simple observations about some properties about the collateral.

\begin{observation}
The price of the collateral in fiat at any queried time point never exceeds $2$M.
\end{observation}
\begin{proof}
This simply follows from the fact that $\pio$ always triggers the collateral invariant condition when a queried BTC price $p_i > p_0$. 
\end{proof}

The next observation is a useful bound on the maximum loss in the price of a unit of Bitcoin in the case where the price of the collateral is between $2$ and $1$.

\begin{observation}\label{obs:alpha_bound}
When $2 > y\cdot p_i > 1$, $\alpha>-0.5$. 
\end{observation}
\begin{proof}
$\alpha = \frac{y}{2}p_t - 1$ and since $\frac{y}{2}p_t > 0.5$, $\alpha>-0.5$.
\end{proof}

\begin{observation}
The price of the collateral in fiat $y \cdot p_i$ hits $1$M when $p_i = \frac{p_0}{2}$. 
This implies that the lender liquidation threshold should be $p_i=\frac{p_0}{2}$ to ensure that the price of the collateral is always larger than the loan, and hence not worthless.
\end{observation}

\paragraph{Game tree and honest strategy.} 
\begin{figure}[htb!]
    \centering
    \includegraphics[width=\textwidth]{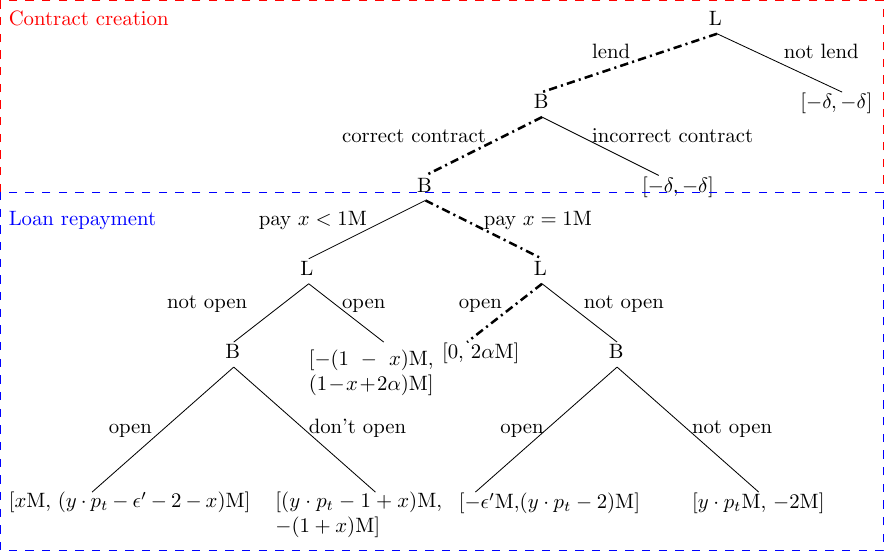}
    \caption{Game tree induced by the protocol with price oracle when the terminal exchange rate $p_t<p_0$ showing the actions the borrower and lender can take at each step of the protocol. 
    The first element of the utility vector (in fiat) at each leaf node is corresponds to the utility of the lender, and the second element the utility of the borrower. 
    The dotted and dashed path depicts the honest strategy.}
    \label{fig:fluctuating_game_tree}
\end{figure}

The $5$ stage extensive form game $\gammao$ induced by $\pio$ is similar to the game induced by $\pib$.
However, the expected payoffs for the players differ depending on the terminal price of BTC.
When the terminal price of BTC $p_t \geq p_0$, the game tree has almost exactly the same payoffs as $\pib$, with the only difference being the payoff of the borrower increases by the excess amount of coins in the temporary account in all leaves stemming from the branch where the borrower decides to pay the full principal of $x=1$M.
We depict the game tree corresponding to $\gammao$ when the terminal price $p_t < p_0$, in~\Cref{fig:fluctuating_game_tree}.
We define the honest strategy profile as $\sigma =(\text{lend}, \text{correct contract}, x=1\text{M}, \text{open})$. 

Following~\Cref{rem:delta}, we use $-\delta$ for some $\delta>0$ to denote the opportunity cost for the lender and borrower if the loan process does not begin.
Our next theorem (with proof in~\Cref{app:main_oracle1}) shows that the honest strategy profile of acting in accordance with the protocol is a subgame perfect equilibrium in $\gammao$.

\begin{theorem}\label{thm:main_oracle1}
Assuming $\delta >1$, $\sigma$ is a subgame perfect equilibrium in $\gammao$.
\end{theorem}

We conclude this section by noting that we can reduce the control of the arbiter in $\pio$ over the excess funds by firstly separating the arbiter from the contract creation (thereby hiding the collateral amount from the arbiter), and also using Zcash~\cite{Ben-SassonCG0MTV14} to hide the amount of excess funds in the case where the price of BTC goes up. 

\smallskip{\em Hiding the amount of excess coins from the arbiter}
A potential issue with the loan repayment phase of $\pio$ is that the arbiter has control over the excess funds in the temporary account. 
Although we assume the arbiter is honest, this could unsettle the borrower, especially if the amount of excess funds is large.
Additionally, if we loosen the threat model slightly to accommodate arbiters that are honest until a certain profit threshold and turn rational above this threshold, these arbiters might steal the excess funds if the value of the funds rises above this threshold.
One possibility to resolve this issue is to modify the contract creation and loan repayment phases of the protocol slightly to hide the value of the excess coins from the arbiter. 

In the contract creation phase, we first change the way the excess funds are handled when the collateral invariant condition is triggered.
Instead of transferring excess funds to a temporary account, the smart contract just needs to keep track of the excess funds.
Additionally, the verification of the contract would have to be done using automated smart contract verification tools instead of using the arbiter.
The arbiter would simply have to be notified by both borrower and lender that the contract is valid. 
In doing so, the content of the contract (including the collateral amount) and the excess funds are separated from the arbiter.

If there are still excess funds in the contract at the end of the loan repayment phase, we can use Zcash~\cite{Ben-SassonCG0MTV14} to shield the transaction amount from the arbiter. 
More precisely, the contract would first have to mint a Zcash coin with the value of the excess funds in the contract and with the receiver address of the borrower. 
The contract then creates another spending contract with the Zcash coin. 
The spending contract takes as input a signed single bit from the arbiter which represents whether the borrower repaid the full loan in the loan repayment phase.
If the bit is $1$, the coin is sent to the borrower.
If the bit is $0$, the coin is burned (i.e., converted into another Zcash coin with receiver address being a burner address). 

Due to the fact that the arbiter is not involved in the contract creation phase and the security properties of Zcash, the value of the excess collateral is somewhat hidden from the arbiter. 
We stress that since the arbiter still knows the principal and the current exchange rate, the arbiter might be able to gain a rough estimate the amount of excess collateral based on typical collateral values corresponding to certain principal values. 
If this risk is still too high for the borrower, they can simply split the loan into smaller chunks and set up a different contract for each chunk.
This would, however, incur a higher cost linear in the number of contracts.

\subsection{Reducing the Arbiter's Control over the Cryptocurrencies}
The previous protocol suffer from three drawbacks. The first minor drawback is that although termination on the side of the borrower fast tracks the protocol to the loan repayment phase, the parties might have to wait up to a month in order to get their funds back, which could lead to an increase in hidden costs like opportunity costs.
The second drawback is the reliance on the external price oracle which can make the protocol complicated to implement over some cryptocurrencies.
The last drawback is the amount of control the arbiter has over the cryptocurrencies. 
We stress that this control is necessary if we want to ensure both lender and borrower are incentivised to follow the honest strategy as defined by the protocol.
This is due to the fact that if the price of the collateral rises too much, the lender and borrower both have incentives to deviate from the honest strategy as defined by the protocol: the lender might not want to open the contract, and the borrower might not want to pay the loan. 

Thus, in this section, we design a protocol $\pino$ that address all three shortcomings of $\pio$. 
Firstly, $\pino$ redefines termination for the borrower to be an \emph{immediate} termination of the contract, whereby the collateral is immediately released and split between the lender and borrower. 
We also allow the borrower to terminate and lender to liquidate the collateral at \emph{any point in time}.
This gives the borrower and lender greater flexibility over the immediacy of their funds at the cost of remaining online for the entirety of the loan term.
$\pino$ also does not require an oracle. 
Rather, the lender and borrower have to remain online throughout the duration of the loan in order to decide if they want to liquidate the collateral.
The arbiter decides if the termination or liquidation is reasonable or unreasonable (which we will detail later), and compensates the relevant parties respectively. 
Finally, $\pino$ limits the amount of BTC controlled by the arbiter to no more than $2\epsilon$. 
In doing so, we have to inevitably sacrifice our notion of the ``honest strategy'' as in the previous protocols, as the lender might not want to open the contract.
Nevertheless, we show that even this minimal amount of control of the arbiter, we can still achieve meaningful stable strategies. 

\subsubsection{Protocol Details}
Our protocol $\pino$ consists of only a contract creation and loan repayment phase.

\begin{description}
\item[Contract creation:]\hfill

The contract creation phase follows almost identically to the contract creation phase of $\pib$, with only an additional condition to unlock the collateral, and a modification of the way the collateral is paid out at maturity which we describe below:
\begin{itemize}
    \item (Early termination/liquidation.) Either the borrower or lender inputs a signed opening to the contract. 
    The contract then pays out $(\frac{y}{2}-\epsilon)$ BTC to the borrower, $(\frac{y}{2}-\epsilon)$ BTC to the lender, and $2\epsilon$ BTC to the arbiter.
    \item At the end of the loan term, the borrower inputs a signed opening to the contract. 
    The contract then pays out Bitcoins worth $(\frac{y}{2}-\epsilon)$M fiat to the borrower, Bitcoins worth $(\frac{y}{2}-\epsilon)$M fiat to the lender, and Bitcoins worth $2\epsilon$M fiat to the arbiter.
\end{itemize}

\item[Loan repayment:]\hfill

The loan repayment phase is also almost identical to that of $\pib$, with the only change being that at any point in time the borrower or the lender can trigger the early termination/liquidation condition of the contract.
When that happens, the loan repayment phase is skipped entirely and the arbiter distinguishes between when the termination/liquidation is reasonable (i.e., when the exchange rate moves past a certain threshold, which we will describe in~\Cref{obs:borrower_thres} and~\Cref{obs:lender_thres} in the next section) or unreasonable. 
If the lender initiates the early liquidation, the arbiter first checks (using their own price oracle) if the rate at which the lender liquidates lies within the actual liquidation threshold of the lender. 
If so, the liquidation is considered reasonable and the arbiter gives the $2\epsilon$ BTC paid out from the contract to the lender. 
Similarly, if the borrower initiates the early termination, the arbiter makes the same check and if the rate lies within the liquidation threshold of the borrower, the arbiter gives $\epsilon$ BTC to the borrower.
If either side makes an unreasonable termination/liquidation, the arbiter keeps the $\epsilon$ BTC from the terminating/liquidating party and gives $\epsilon$ BTC to the other party.
\end{description}

We defer the full security analysis of $\pino$ to~\Cref{app:security_oracleless}.
As a brief overview, we show that although players have a wider range of action profiles to choose from in the game $\gammano$ induced by $\pino$, we can still achieve meaningful stable strategies if we make two simplifying assumptions.
Firstly, we restrict the class of players to be non-myopic, that is, the expected utility given a payoff of $x$ is the same when $x$ is paid out at time $i$ and some other time $t>i$.
Secondly, we assume the expected price of BTC at time $t$ is equal to the current BTC price at time $i$ for $t>i$.
With these assumptions, we can show that as long as the price of BTC is not too high, following the protocol is a subgame perfect equilibrium of the game induced by $\pino$ (refer to~\Cref{thm:main_oracleless} in~\Cref{app:security_oracleless} for more details).

\section{Discussion and Extensions}

\subsection{Comparison of our Protocols}\label{sec:comparison}
We introduce the following metrics to compare the protocols $\pio$ and $\pino$: (1) the amount of funds the arbiter controls, (2) how often parties have to remain online, (3) immediacy of the funds, (4) complexity of the protocol.

$\pino$ fares better than $\pio$ in terms of the amount of control the arbiter has over the locked funds. 
In $\pino$, the arbiter only controls $2\epsilon$ of the collateral and this is independent of the value of the collateral, whereas in $\pio$ the arbiter can potentially control $(y \cdot p_i -2)$M amount of funds which can be arbitrarily large depending on the price of BTC $p_i$.

In $\pino$, parties have to remain online through the duration of the loan to monitor the exchange rates and decide if they are in a profitable position to liquidate or terminate the protocol, whereas in $\pio$ parties do not need to remain online as they can specify liquidation and termination thresholds during contract creation. 
On the flip side, $\pino$ also allows the lender and borrower to have more control over the immediacy of their funds, as opposed to waiting until a query returns a price that fulfils the threshold as in $\pio$.

Finally, we note that whilst $\pio$ allows collateral to be redistributed according to the current price of BTC which allows more fine-grained mechanism design, these oracle calls and collateral redistribution make $\pio$ more complex compared to $\pino$. 

\smallskip{\em Example usage scenarios.} 
We envision $\pio$ to be more popular with borrowers and lenders that do not have the time to be online and track the movement of BTC, or if they more indifferent to the precise liquidation threshold but simply desire a small amount of liquidity. 
Examples of such are ordinary users with some spare cryptocurrencies who actually want to take a loan but prefer to remain offline most of the time. 
Another class of users that $\pio$ could be suited for would be borrowers that need microloans and thus might not be too worried about the amount of control the arbiter has over any excess collateral.
We envision $\pino$ to be more popular with traders or parties for which tracking CeFi and DeFi markets are part of their daily routine, and thus they could easily view making a loan using this protocol as investing in another type of financial derivative.
Hence, the flexibility on when to liquidate the collateral and the immediacy of the funds would be more important.

\subsection{Extensions and Interesting Connections}\label{sec:extensions}

\smallskip{\em Bitcoin compatibility.}
$\pib$ is Bitcoin compatible, as the contract only needs to lock the collateral and release at specified times corresponding to specified signatures in the future. 
This can be implemented using nested conditional statements and Timelocks\footnote{https://github.com/bitcoin/bips/blob/master/bip-0065.mediawiki/}.
We use $pk_L, pk_B, pk_A$ to denote the public keys of a public key signature scheme of the lender, borrower, and arbiter respectively. Details on the UTXO model can be found in~\Cref{app:bitcoin}.

~\Cref{fig:basictx} in~\Cref{app:bitcoin} depicts the transaction with the spending conditions of the collateral for $\pib$. The overall collateral amount of $2M$ is split into 3 outputs, the first two of value $(1-\epsilon)$M, and the last output of value $2\epsilon$M. 
For each output, the first spending condition corresponds to the lender opening the contract and releasing the collateral to the borrower ($s_L$ denotes the lender's signed opening to the contract). 
The second spending condition for each output corresponds to the borrower opening the contract. 
This differs for each output. 
The first output of $(1-\epsilon)$M can only be spent by the borrower, the second output of $(1-\epsilon)$M can only be spent by the lender, and the final output of $2\epsilon$M can only be spent by the arbiter.
The last spending condition for each output corresponds to the release of the collateral to the lender after the loan repayment phase is over (i.e., after $t\geq t^*+13$ as in the timeline in~\Cref{fig:protocol_timeline}).

$\pio$ can also be implemented over Bitcoin as we can use discrete log contracts~\cite{dlc,GuillyKK22,MadathilTVMFM2022} during the loan term to access price information from an oracle.
To implement collateral redistribution, the collateral is split into $k$ output chunks of value $\frac{y}{k}$ each and chunks are given to the arbiter's temporary account during the redistribution process such that the remaining chunks have value $\leq2$M. 

We leave the optimal selection of the parameter $k$ up to the users, but we stress that a larger $k$ parameter would ensure that the total value of collateral outputs in the main smart contract is closer to $2$M. 
Indeed, the larger the value of $k$, the closer the incentives of the lender would be to the basic protocol. 
Additionally, this prevents the lender liquidation condition from being prematurely triggered due to insufficient collateral.
We note though, that a larger $k$ parameter would increase the size of the smart contract and incur higher verification and computational cost.

If we want to make the modification to hide the excess funds from the arbiter Bitcoin compatible, we simply add an extra condition in the contract creation phase that stipulates excess funds would be minted into a new coin with a receiver address belonging to the borrower or burned depending on the arbiter's input. 

Finally, the contract in $\pino$ is similar to the one in $\pib$ and thus also Bitcoin compatible.
We provide more details in~\Cref{app:bitcoin}.

\smallskip{\em Interest rates.}
Our protocols do not take into account interest rates and their impact on the protocol. 
Nevertheless, we note that accounting for interest rates would have minimal impact on our analysis.
The main changes we would need to account for would be the opportunity cost for the borrower and lender and thus their termination/liquidation thresholds, as well as the initial size of the collateral.
We leave this as an interesting direction of future work.

\section{Conclusion}
In this work, we present the first limited-custodial protocols for cryptocurrency-backed loans in the mixed CeFi-DeFi setting.
Unlike previous protocols in this setting, our protocols limit the control of the trusted third party arbiter to only a small fraction of the collateral and are also capable of dealing with fluctuating exchange rates between the fiat and cryptocurrency.
We also compare our two protocols $\pio$ and $\pino$, showcase their relative advantages and shortcomings, and highlight various user settings where it would make sense for a user to pick one protocol over the other, which shows that our protocols are adaptable to various scenarios. 
Finally, we provide game theoretic analysis of our protocols, showing that in all our protocols the strategy as specified by the protocol is a subgame perfect equilibrium, which means that rational Borrowers and Lenders do not have any incentive to deviate from our protocols.

\bibliography{references}

\appendix

\section{Game theoretic concepts and subgame perfect equilibrium}\label{app:gtconcepts}

Let $\Gamma = (N, (A_i), (u_i))$ be an $N$ player game where $A_i$ is a finite set of actions for each player $i \in [N]$ and denote by $A := A_1 \times \cdots \times A_N$ the set of action profiles. 
The utility function of each player $i$, $u_i: A \rightarrow \mathbb{R}$, gives the payoff player $i$ gets when an action profile $a \in A$ is played. A \emph{strategy} $\sigma_i \in \Delta(A_i)$ of a player $i \in [N]$ is a distribution over all possible actions of the player.

\begin{definition} (Nash Equilibrium).
A Nash Equilibrium (NE) of $\Gamma$ is a product distribution $\alpha \in \times_{j \in [N]} \Delta(A_j)$ such that for every player $i \in [N]$ and for all $a'_i$ in $A_i$, 
$$\mathbb{E}_{a \leftarrow \alpha}[u_i(a)] \geq \mathbb{E}_{a \leftarrow \alpha}[u_i(a'_i, a_{-i})] $$
\end{definition}

Intuitively, a vector of strategies $\alpha$ is a Nash Equilibrium if no unilateral deviation fron $\alpha$ can strictly increase the utility of any player.



Nevertheless, we note that the solution concept of Nash equilibrium only applies to single-shot games, which are too restrictive to model the entire set of actions players can embark on over time in our loan setting. 
Multi-round games where players' actions arrive sequentially are modelled as extensive-form games. 
For a formal definition of extensive-form games, see, e.g.,~\cite{osbourne04}. 
For our purposes, however, one can simply think of extensive-form games as defined by a game tree $T$. 
All non-leaf vertices in $T$ are partitioned into sets with each set corresponding to one player in the game. 
A player move (or action) at vertex $x$ in their vertex set is simply a choice of an edge from $x$ to some child of $x$. 
A path from the root of $T$ to a leaf vertex corresponds to a sequence of player moves made by the players in the game. 
Each leaf of $T$ is labelled with a utility vector which shows how much utility each player gets when the game play terminates at this leaf. 
In the imperfect information or simultaneous action settings, the vertices belonging to any player are further partitioned into information sets $I$ which capture the idea that a player making a move at any vertex $x \in I$ is uncertain whether they are making the move from $x$ or any other vertex $x' \in I$. 

A subgame of an extensive-form game corresponds to a subtree in $T$ rooted at any non-leaf vertex $x$ that belongs to its own information set $I$, i.e., there are no other vertices that are in $I$ except for $x$. 

\begin{definition} (Subgame perfection).
Let $\Gamma$ be an extensive form game.
A strategy profile is a \emph{subgame perfect equilibrium} of $\Gamma$ if it is a Nash equilibrium for all subgames in $\Gamma$. 
\end{definition}

\section{Proof of Theorem~\ref{thm:basic}}\label{app:proof_basic}

\begin{proof}
We proceed by backwards induction, starting with the actions of the borrower at the subgames at stage $5$ of $\gammab$.
We first note that we assume $\epsilon < 1$, thus the expected payoff corresponding to the borrower in the subgames corresponding to both subtrees is always larger when the borrower opens the contract than when the borrower does not open the contract. 
Thus, in both subgames, the pure strategy of opening the contract dominates all other pure and mixed strategies, resulting in the optimal expected payoff of [$0$, $-(x+\epsilon)$M] in the left subgame and [$-\epsilon$M, $0$] in the right subgame.

Now we analyse the payoffs of the lender strategies for the subgames at stage $4$ of $\gammab$. 
We note that in the subgame corresponding to the left subtree, the lender's expected payoff of not opening the contract is $0 > -(1-x)$M which is the expected payoff corresponding to opening the contract.
In the subgame corresponding to the right subtree, the lender's expected payoff of opening the contract is $0 > -\epsilon$M which is the expected payoff of not opening the contract. 
Thus, the optimal strategy for the lender is pure strategy of not opening the contract in the left subgame, and opening the contract in the right subgame. 

In stage $3$ of $\gammab$, we first note that choosing $x=0$ maximises the expected payoff of the borrower which is $-x-\epsilon$ if the borrower chooses to pay $x<1$. 
However, even when $x=0$, this is strictly less than the payoff corresponding to paying $x=1$M which is $0$.
Thus, the optimal strategy for the borrower is the pure strategy of paying $x=1$M. 

Finally, we note that the utility vectors corresponding to any strategy at both stages $2$ and $1$ are exactly the same, and so we simply pick the pure strategies ``lend" and ``correct contract" to be the strategies corresponding to the lender and the borrower in stages $1$ and $2$. 
As such, the honest strategy profile $\sigma =(\text{lend}, \text{correct contract}, x=1\text{M}, \text{open})$ is a subgame perfect equilibrium of $\gammab$. 
\end{proof}

\section{Bitcoin compatibility details}\label{app:bitcoin}

\smallskip{\em UTXO model details.}
In the \emph{unspent transaction output} (UTXO) model, each unit of currency (which we will hereafter call coins) exists as an output of a transaction which contains two fields: $(v, c)$. 
The first field of the output $v$ is the value of the output, and the second field $c$ contains spending conditions for the output.
The expressivity of spending conditions depend on the underlying scripting language of the blockchain and common examples are single-user and multi-user ownership, time locks, and hash locks.
The users that are entitled to spend the output also provide a list of witnesses that fulfill the spending conditions of the output. 
Outputs are transferred from user to user via \emph{transactions}, which take as input a set of transaction outputs, and outputs another set of transaction outputs together with spending conditions. 
Note that the total value of the transaction inputs must equal that of the transaction outputs.

\begin{figure}[htb!]
    \centering
    \includegraphics[width=0.5\textwidth]{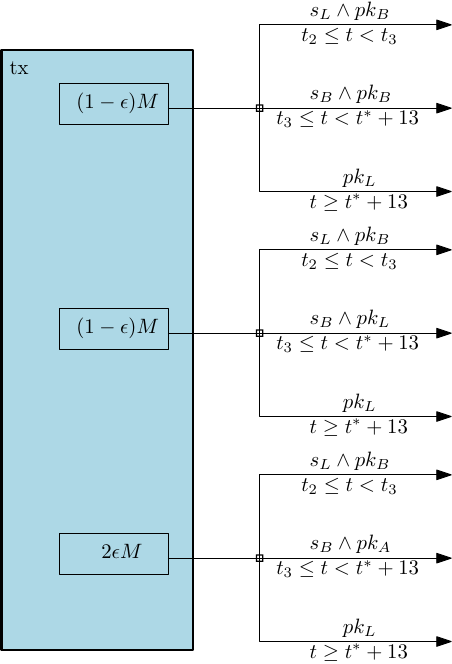}
    \caption{Transaction depicting the release of collateral in $\pib$.}
    \label{fig:basictx}
\end{figure}

\smallskip{\em Details of collateral redistribution for $\pio$.}
We now detail the collateral redistribution in $\pio$. 
Note that apart from collateral redistribution, the rest of the collateral release conditions in the protocol is similar to that in $\pib$.

Recall that the initial collateral amount $y$ should equal $2$M fiat when converted using the exchange rate $r_0$ at the start of the protocol. 
To implement collateral redistribution, the collateral in the smart contract is split into $k\geq2$ output chunks of value $\frac{y}{k}$ each. 
Upon querying the oracle for the price $p_i$ of the cryptocurrency and if $p_i>p_0$, the contract selects output chunks to send to the arbiter's temporary account such that the remaining chunks have value $\leq2$M. 
The collateral in the temporary account is also controlled by a secondary smart contract.
Therefore, if the price $p_i <p_0$ conversely, the secondary smart contract holding the collateral in the temporary account transfers sufficiently many output chunks such that the total value of the outputs in the main smart contract $\leq2$M.

We stress that due to the condition that the total collateral stored in the main contract is $\leq 2$M, the security of the protocol still holds from clause $5$ in~\Cref{sec:oracledetails} and~\Cref{thm:main_oracle1}.

\section{Proof of Theorem~\ref{thm:main_oracle1}}\label{app:main_oracle1}
\begin{proof}
Let us denote the event that an execution contains a query that returns a price $p_i$ such that $P_i$ is smaller than the liquidation threshold $\frac{p_0}{2}$ as $E_L$.
Similarly, we denote the case where an execution contains a query that returns a price $p_i$ such that $p_i > \tau$ as $E_B$.
We first analyse the case where neither $E_L$ nor $E_B$ occur, and only the case where $p_t < p_0$ as the case of $p_t \geq p_0$ follows exactly from the proof of~\Cref{thm:basic}.
We proceed by backwards induction starting from stage $5$ of $\gammao$.

The expected payoff of the borrower in the subgame corresponding to the left subtree at stage $5$ of $\gammao$ when opening the contract is always larger than the expected payoff when not opening the contract since $y \cdot p_t>1 \implies \exists \epsilon' >0$ s.t. $y\cdot p_t-1-\epsilon' > 0$. 
Since $y\cdot p_t>1$, the expected payoff of the borrower in the subgame corresponding to the right subtree when opening is $>-1$M which is larger than the expected payoff when not opening of $-2$M.
As such, the optimal strategy for the borrower for both subgames is to open the contract.
We can then eliminate these subgames and replace them with the optimal payoffs of $[x\text{M}, (y \cdot p_t-\epsilon'-2-x)\text{M}]$ on the left reduced subtree and $[-\epsilon'\text{M}, (y\cdot p_t -2)\text{M}]$ on the right reduced subtree.

Now we analyse the expected payoffs of the lender at stage $4$ of $\gammao$. 
In the subgame corresponding to the left subtree, the expected payoff of the lender is $x$M if the lender does not open the contract. 
This is always larger than the expected payoff of $-(1-x)$M of opening the contract.
In the right subtree, the expected payoff of opening is $0$ which is strictly larger than that of not opening which is $-\epsilon'$.
Again, we eliminate these subgames and replace them with the optimal payoffs of $[x\text{M}, (y \cdot p_t-\epsilon'-2-x)\text{M}]$ on the left reduced subtree, and $[0, 2\alpha\text{M}]$ on the right reduced subtree. 

In stage $3$ of $\gammao$, we first observe that choosing $x=0$ maximises the expected payoff of the borrower.
However, even when $x=0$, the expected payoff of the borrower when paying $x=0$ is $y\cdot p_t-2-\epsilon'$ which is still strictly less than $2\alpha = y\cdot p_t-2$ which corresponds to the expected payoff when paying the full principal of $x=1$M.
Thus, we can eliminate this subgame and replace it with the optimal payoff vector $[0, 2\alpha\text{M}]$. 

For stages $1$ and $2$, we note that since we assume $\delta >1$ and since $2\alpha > -1$ from~\cref{obs:alpha_bound}, creating the correct contract at stage $2$ and lending the loan at stage $1$ are the pure strategies that would maximise the payoffs for the borrower and lender respectively.  

Now we analyse the case where $E_L$ occurs (i.e., a queried price returns $p_i<\frac{p_0}{2}$ for some $i\in [q]$). 
Given $E_L$ occurs, we simply need to show that the strategy profile beginning with prefix (lend,correct contract) dominates any other strategy profile in this case. 
We observe that the moment $E_L$ occurs, the full collateral of $y$ is transferred to the lender at the price $p_i$ and since $p_i<\frac{p_0}{2}$, $y\cdot p_i<1$M. 
This would lead to an expected payoff of $y\cdot p_i-1$ for the lender. 
The borrower gains the principal but loses the collateral, which leads to an expected payoff of $1-y \cdot p_i>0$. 
Since we assume $\delta > 1$ and the expected payoffs for both parties when not starting the loan process is $-\delta$ each, we see that even in the case where $E_L$ occurs, the expected payoffs of both parties are still larger than the payoffs of not beginning the loan process.
Thus, the strategy profile beginning with prefix (lend,correct contract) would always lead to a larger expected payoff. 

For the case of $E_B$, (i.e., a queried price returns $p_i>\tau$ for some $i\in [q]$), we also need to show that the strategy profile beginning with prefix (lend,correct contract) dominates any other strategy profile. 
However, since $\tau > p_0$, this initiates the loan repayment process at the current price $p_i > p_0$.
From the proof of~\Cref{thm:basic}, we see that the optimal payoff from the eliminating other subgames at stage $3$ is $[0,c]$ where $c>0$ is the amount of excess coins paid to the borrower from the temporary account.
Since both lender and borrower payoffs are larger than the payoffs of not beginning the loan process ($-\delta$ each), the strategy profile beginning with prefix (lend,correct contract) would always lead to a larger payoff. 

As such, the honest strategy profile $\sigma =(\text{lend}, \text{correct contract}, x=1\text{M}, \text{open})$ is a subgame perfect equilibrium of $\gammao$.
\end{proof}

\section{Security analysis of $\pino$}\label{app:security_oracleless}
We first make $2$ simple observations on the termination/liquidation threshold for the lender and borrower.
Let $p_i$ denote the current price of BTC.

\begin{observation}\label{obs:borrower_thres}
(Borrower termination threshold.)
\label{lem:borrower_thres}
A rational borrower will not lose out on the collateral amount if termination happens when $p_i$ hits the threshold $\rho_B := \frac{p_0}{1-\epsilon p_0}$.
\end{observation}
\begin{proof}
The payoff of the borrower when the borrower terminates the protocol when the price of BTC is $p_i$ is the sum of the principal and the coins given out from the contract, i.e., $1+(\frac{y}{2}-\epsilon)p_i$. 
We note that this expression is exactly $2$ when $p_i=\frac{p_0}{1-\epsilon p_0}$, and only increases when $p_i > \frac{p_0}{1-\epsilon p_0}$. 
\end{proof}

\begin{observation} (Lender liquidation threshold.)
\label{obs:lender_thres}
A rational lender will not lose out on the principal if liquidation happens when $p_i$ hits the threshold $\rho_L := \frac{p_0}{1+\epsilon p_0}$.
\end{observation}
\begin{proof}
If a lender liquidates and it is reasonable, the lender will get $(\frac{y}{2}+\epsilon)p_i$M ($y/2-\epsilon$ BTC from the contract and $2\epsilon$ BTC from the arbiter). 
To protect the lender against falling BTC prices, the authorisation threshold should be such that the lender cannot lose out on the full principal of $1$M during liquidation, thus $(\frac{y}{2}+\epsilon)p_i =1 \implies p_i = \frac{p_0}{1-\epsilon p_0}.$ 
\end{proof}

\Cref{obs:borrower_thres} implies that it is reasonable for the borrower to terminate the protocol at any point in time where the price of BTC rises past $\rho_B = \frac{p_0}{1-\epsilon p_0}$ as any price point past this threshold would be a net gain for the borrower. 
Thus, we say that the borrower's termination is reasonable if the borrower terminates at some time point $i$ and $p_i \geq \rho_B$.
Otherwise, the termination is considered unreasonable.
Likewise,~\Cref{obs:lender_thres} implies that the lender should liquidate the collateral at $\rho_L$ as this is where the lender does not lose out on the principal. 
Thus, we way the lender's liquidation is reasonable if $p_i \leq \rho_L$, otherwise, the liquidation is considered unreasonable.

A simple corollary of~\Cref{obs:lender_thres} is the following inequality which upper bounds the instant payoff in fiat from an unreasonable termination of the borrower when the price of BTC is not too low, with the smallest possible final payoff the borrower can get from the collateral at the end of the loan repayment phase, assuming the lender also liquidates when $p_t \leq \rho_L$.

\begin{corollary}
\label{cor:lender_thresh}
If $\rho_L < p_i <\rho_B$, then $1 + (\frac{y}{2}-\epsilon)p_i < y \cdot \rho_L$.
\end{corollary}
\begin{proof}
It is easier to work with exchange rates instead of price here. 
Recall that $y=2r_0$. 
Thus, $1 + \frac{y/2 -\epsilon}{r_i} < \frac{y}{r_0+\epsilon}$ can be rewritten as  $(r_0+\epsilon)(r_i + r_0-\epsilon) < 2r_ir_0$.
Rearranging, we see that this holds when $r_i < r_0 + \epsilon$, or equivalently, when $p_i> \rho_L$ 
\end{proof}

\smallskip{\em Liquidation game.}
The game $\gammano$ induced by $\pino$ is similar to the game $\gammao$ just with the addition of a new liquidation game that can be played once the protocol begins (i.e., in between stages $2$ and $3$ of the main game $\gammano$).
We depict the full game tree corresponding to $\gammano$ in~\Cref{fig:oracleless_game_tree}.
$\gammano$ is a $6$ stage extensive form game with the extra stage being the liquidation game that lasts for the entirety of the maturity period.
The liquidation game begins with a move by nature that samples the price of BTC $p_i$ at any time step.
The two players (lender and borrower) can then choose between actions $\{$liquidate, not liquidate$\}$, which represents the ability of either player to decide at any point whether to liquidate (or terminate in the case of the borrower) the collateral and, in doing so, terminate the overall game $\gammano$ prematurely. 
Note that the liquidation game is played continuously (i.e., a new rate is sampled from the environment and the game repeats itself) in the event where neither player chooses to liquidate the collateral until the loan maturity date.
If neither player chooses to liquidate the collateral by the end of the liquidation game, the game proceeds to the next stage which is the loan repayment phase.

\begin{figure}[htb!]
    \centering
    \includegraphics[width=\textwidth]{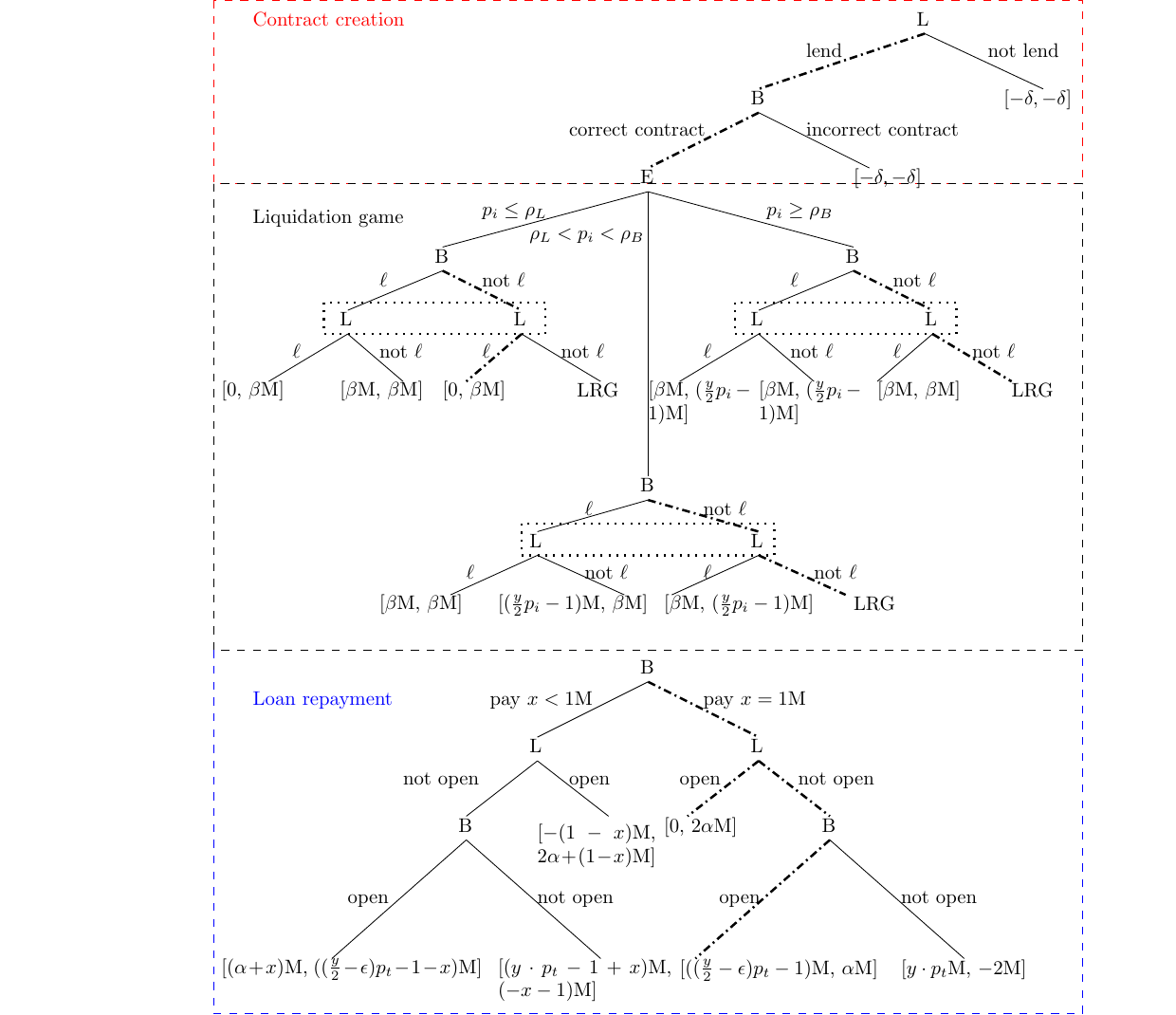}
    \caption{Game tree induced by the protocol without oracle showing the actions the borrower and lender can take at each step of the protocol. 
    In the liquidation game, the vertex $E$ represents a move from nature and $\ell$ denotes liquidation. 
    The dotted rectangle around the lender vertices in the liquidation game denote that they belong in the same information set.
    For ease of presentation, we use $\beta$ to denote $(\frac{y}{2}-\epsilon)p_i-1$.
    LRG at the leaf nodes in the liquidation game the expected payoff from the loan repayment game phase.}
    \label{fig:oracleless_game_tree}
\end{figure}

\smallskip{\em Movement of exchange rates.}
In order to analyse stable strategy profiles for the liquidation game, we need to compare the instant payoff from liquidating the collateral (i.e., the amount paid out when liquidating the collateral converted to fiat at the current exchange rate) compared to the future payoff if the game continues all the way through the loan repayment phase.
As such, we make a simplifying assumption that expected terminal price of BTC should be the same as $p_i$, that is, $\mathbb{E}[p_t] = p_i$.
We also assume the players are non-myopic, that is, the expected utility given a payoff of $x \in \mathbb{R}$ is the same when $x$ is paid out at time $i$ and some time $t > i$.

\smallskip{\em Strategy profiles.} 
Let us denote the event that the price of BTC hits or drops beneath $\rho_L$ as $E_L$, and the event that the price of BTC hits or exceeds $\rho_B$ as $E_B$.
We use the first element in the strategy profile for the liquidation game to represent the lender's strategy in the game, and the second element to represent the borrower's strategy. 
Let the strategy profile $\sigma$ (dotted and dashed path in black in~\Cref{fig:oracleless_game_tree}) be defined as

\begin{itemize}
    \item (\text{lend}, \text{correct contract}, (\text{not liquidate},\text{not liquidate}), x=1\text{M}, \text{open}) if an execution does not contain $E_L$ and $E_B$
    \item (\text{lend}, \text{correct contract}, (\text{liquidate},\text{not liquidate})) if an execution contains $E_L$
    \item If an execution contains $E_B$
    \begin{itemize}
        \item $p_t < \rho_B$: (\text{lend}, \text{correct contract}, (\text{not liquidate},\text{not liquidate})
        \item $p_t \geq \rho_B$: (\text{lend}, \text{correct contract}, (\text{not liquidate},\text{not liquidate}), x=1\text{M}, \text{not open}, \text{open})
    \end{itemize}
\end{itemize}

In the next theorem, we show that for non-myopic players, the strategy $\sigma$ as defined above is a subgame perfect equilibrium of $\gammano$.

\begin{theorem}\label{thm:main_oracleless}
If we set $\epsilon < \frac{1}{2p_0}$, and $\delta>2(\rho_L \cdot \frac{y}{2}-1)$, then $\sigma$ is a subgame perfect equilibrium in $\gammano$ for non-myopic players.
\end{theorem}

To prove~\Cref{thm:main_oracleless}, we split it into the three cases where $E_L$ and $E_B$ do not occur, when $E_L$ occurs, and when $E_B$ occurs.

\begin{lemma}\label{lem:honest_noliquid}
If we set $\epsilon < \frac{1}{2p_0}$, and $\delta>2(\rho_L \cdot \frac{y}{2}-1)$, then $\sigma$ is a subgame perfect equilibrium in $\gammano$ given an execution without $E_L$ and $E_B$ .
\end{lemma}
\begin{proof}
We start with analysing the payoffs of the borrower at the right subtree at stage $6$ of $\gammano$. 
If the borrower opens the contract the borrower will get a payoff of $\alpha = \frac{y}{2}p_t -1$M compared to a payoff of $-y\cdot p_t$M of not opening. 
A large value of $p_t$ only makes $\alpha$ larger than $-y\cdot p_t$, so we simply need to check the case where $p_t$ is small.
But since $\epsilon < \frac{1}{2p_0}$ can be rewritten as $\epsilon < \frac{r_0}{2} \implies 4(r_0+\epsilon) < 4r_t <3y \implies \alpha = \frac{y}{2r_t}-1 > -\frac{y}{r_t}$.
Thus, the optimal strategy for the borrower in the right subtree is to open the contract. 
In the left subtree, opening the contract gives a payoff of $((\frac{y}{2}-\epsilon)p_t-x-1)$M which is always larger than not opening the contract which gives a payoff of $(-1-x)$M as $(\frac{y}{2}-\epsilon)p_t>0$.
Thus, the optimal strategy for the borrower is also to open the contract in this case.
We can then eliminate both subgames corresponding to these subtrees and replace them by the optimal utility vectors of $[(\alpha+x)\text{M}, ((\frac{y}{2}-\epsilon)p_t-x-1)\text{M}]$ for the left subgame and $[((\frac{y}{2}-\epsilon)p_t-1)\text{M}, \alpha\text{M}]$ for the right subgame.

Next, we analyse the strategies of the lender at stage $5$ of $\gammano$, starting with the subgame corresponding to the subtree on the right.
If the lender does not open the contract, the lender gets a payoff of $(\frac{y}{2}-\epsilon)p_t-1$ which will always be $<0$ for small $p_t$. 
For large $p_t$, $(\frac{y}{2}-\epsilon)p_t-1 < \frac{y}{2}p_0-1 =0$ since we assume $E_B$ does not happen, and thus the optimal strategy for the lender would be to open the contract in this subgame.
For the subgame corresponding to the left subtree, the payoff of the lender when not opening is $(\alpha+x=\frac{y}{2}p_t-1+x)$M which is always larger than opening which gives a payoff of $(-1+x)$M since $\frac{y}{2}p_t>0$.
As such, the optimal strategy for the lender would be to not open the contract in this subgame.
We can then eliminate both subgames corresponding to these subtrees and replace them by the optimal utility vectors of $[(\alpha+x)\text{M}, ((\frac{y}{2}-\epsilon)p_t-x-1)\text{M}]$ for the left subgame and $[0, 2\alpha\text{M}]$ for the right subgame.

We now analyse the strategies of the borrower at stage $4$ of $\gammano$. 
We note that by setting $x=0$, the borrower maximises their payoff of $(\frac{y}{2}-\epsilon)p_t-x-1$ when choosing to pay $x<1$M.
However, even when $x=0$, the borrower's payoff when not paying the full amount is $\frac{y/2-\epsilon}{r_t} -1< 2\alpha$ which is the payoff of the borrower when paying the full principal since we are in the setting where $p_t > \rho_L$.
As such, the optimal strategy for the borrower is to pay the full principal, and we can eliminate the subgame corresponding to this subtree and replace it with the optimal utility vector of $[0, 2\alpha\text{M}]$.

Now we analyse the strategies of both players at stage $3$, the liquidation game (middle subtree in the liquidation game phase in~\Cref{fig:oracleless_game_tree}).
We want to ensure that either player does not gain from deviating from the protocol and performing an unreasonable liquidation of the collateral at any time point between the start of the loan $t^*$ and the maturity.
We first analyse the expected payoffs of the lender. 
If the borrower chooses not to liquidate, the lender's expected payoff is $(\frac{y}{2}-\epsilon)p_i-1$ if the lender liquidates which is the same as the optimal expected payoff of $(\frac{y}{2}-\epsilon)p_t-1$ when proceeding to the loan repayment game phase, as we assume $\mathbb{E}[p_t] = p_i$.
If the lender chooses not to liquidate,  from~\Cref{cor:lender_thresh} we observe that the borrower's expected payoff is $1+ (\frac{y}{2}-\epsilon)p_i -2$ when liquidating which is always smaller than the optimal expected payoff of $2\alpha = y\cdot p_t -2$ of not liquidating.
Thus, any unilateral deviation from the strategy profile $(\text{not liquidate}, \text{not liquidate})$ does not increase the expected payoff of either party in the liquidation game. 
We can thus eliminate this subgame and pick $\{\text{not liquidate}, \text{not liquidate}\}$ as the Nash equilibrium.

At stages $2$ and $1$ of $\gammano$, since we assume $\delta>2(\rho_L \cdot \frac{y}{2}-1) \geq 2\alpha$, both borrower and lender would rather opt to start the loan process than terminate the protocol before the start of the loan process. 

As such, the strategy profile $\sigma =(\text{lend}, \text{correct contract}, \text{not liquidate}, x=1\text{M}, \text{open})$ is a subgame perfect equilibrium of $\gammano$ given an execution without $E_L$ and $E_B$.
\end{proof}

\begin{lemma}\label{lem:subgame_el}
If we set $\epsilon < \frac{1}{2p_0}$, and $\delta>2(\rho_L \cdot \frac{y}{2}-1)$, then $\sigma \restriction_{E_L}$ is a subgame perfect equilibrium in $\gammano$ given an execution with $E_L$.
\end{lemma}

\begin{proof}
Suppose an execution contains an occurrence of $E_L$.
Denote the price of BTC when $E_L$ occurs as $p_i$.
If $E_L$ occurs and the borrower does not liquidate, from~\Cref{obs:lender_thres}, the lender gets $1$M when liquidating at this point, which means that the overall payoff of the lender is $0$.
This is the same as the optimal expected payoff of the lender if the lender does not liquidate as the lender will simply proceed to the loan repayment game of $\gammano$ where the subgames at stages $6$, $5$, and $4$ have been eliminated and replaced by the optimal utility vector $[0, 2\alpha]$M.  
If the lender chooses to liquidate, the expected payoff of the borrower is $(\frac{y}{2}-\epsilon)p_i$ for both choices of liquidating and not liquidating.
Since no unilateral deviation from the strategy profile $(\text{liquidate}, \text{not liquidate})$ gives either player a higher expected utility, we choose $(\text{liquidate}, \text{not liquidate})$ to be the Nash Equilibrium of this subgame.
The analysis of the strategies for the first $2$ stages of the game follows identically as in the proof of~\Cref{lem:honest_noliquid}.
Thus, the strategy profile $\sigma \restriction_{E_L} =(\text{lend}, \text{correct contract}, (\text{liquidate}, \text{not liquidate}))$ is a subgame perfect equilibrium in $\gammano$. 
\end{proof}

\begin{lemma}
\label{lem:subgame_eb}
If we set $\epsilon < \frac{1}{2p_0}$, and $\delta>2(\rho_L \cdot \frac{y}{2}-1)$, then $\sigma \restriction_{E_B}$ is a subgame perfect equilibrium in $\gammano$ given an execution with $E_B$.
\end{lemma}
\begin{proof}
Suppose an execution contains an occurrence of $E_B$.
Denote the price of BTC when $E_B$ occurs as $p_i$.
If $E_B$ occurs, $p_i$ is $\rho_B$, and the borrower gets $1+(\frac{y}{2}-\epsilon)p_i + \epsilon p_i = 2 + \epsilon p_i $M when liquidating at this point, which only increases as $p_i$ increases.

We first show that, if the lender chooses not to liquidate, the expected utility of the borrower is smaller when liquidating than when not liquidating.
If the borrower chooses not to liquidate and $E_L$ does not happen as well, the game proceeds to the loan repayment stage, where there are two outcomes that could happen, depending on the terminal BTC price $p_t$.
If the terminal price $p_t < \rho_B$, the analysis of the optimal strategies of the lender and receiver would follow as per the setting without an execution of $E_L$ and $E_B$. 
Thus, we need to compare the payoff of the borrower when liquidating to the optimal payoff vector $[0,2\alpha]$M after eliminating the subgames at stages $6$, $5$, and $4$.
From~\Cref{cor:lender_thresh}, since $p_t$ is larger than $\rho_L$ (the lender will liquidate otherwise), the amount the borrower gets when liquidating the collateral will always be less than the amount paid out by the contract at the end of the loan of. 

If $p_t \geq \rho_B$, the optimal strategy that maximises the expected payoff of the borrower at stage $6$ of $\gammano$ is also to open the contract. 
However, the optimal strategy of the lender at stage $5$ of $\gammano$ changes in the subgame corresponding to the right subtree. 
The expected payoff of the lender when choosing not to open the contract is now $(\frac{y}{2}-\epsilon)p_t -1 \geq 0$ if $p_t \geq \rho_B$ and thus the lender would choose not to open the contract in this case. 
The optimal strategy of the lender in the subgame corresponding to the left subtree is unchanged.
Thus, we eliminate both subgames and replace them by the optimal utility vectors of $[(\alpha+x)\text{M}, ((\frac{y}{2}-\epsilon)p_t-x-1)\text{M}]$ for the left subgame and $[((\frac{y}{2}-\epsilon)p_t-1)\text{M}, \alpha\text{M}]$ for the right subgame.
Analysing the borrower's expected payoffs at stage $4$, we see that the strategy that maximises the borrower's payoffs is still choosing to pay the full loan of $x=1$M as $\alpha = \frac{y}{2}p_t -1 > (\frac{y}{2}-\epsilon)p_t -1$. 
As such, we just need to compare the expected payoff of the borrower when liquidating to the optimal payoff of $\alpha$M.
Since $p_i > \rho_B$, the expected payoff of the borrower is $1+ \frac{y}{2}p_i-2 = \frac{y}{2}p_i-1$M and since we assume $\mathbb{E}[p_t] = p_i$, the expected payoff of the borrower is exactly the same as $\alpha$M.
Thus, liquidation gives the same expected payoff as no liquidation in this case. 
As such, if we fix the lender strategy to be not liquidate, given that the expected payoff of the borrower  when liquidating is either same or smaller than the expected payoff of not liquidating, we see that the strategy of not liquidating maximises the expected payoff of the borrower in this case.

Now we need to check if the expected utility of the lender is larger when the lender liquidates if the borrower chooses not to liquidate.
The expected payoff of the lender is $(\frac{y}{2}-\epsilon)p_i -1$ if the lender liquidates when the borrower does not liquidate.
Using the same analysis as the above, if the terminal BTC price $p_t > \rho_B$, the lender will not choose to open the contract and get an expected payoff of $(\frac{y}{2}-\epsilon)p_t -1$ which is the same as the case where the lender does not liquidate and continues to the loan repayment game since we assume $\mathbb{E}[p_t]=p_i$.
If the terminal BTC price $p_t \leq \rho_B$, the lender will choose to open the contract and get a payoff of $0 > (\frac{y}{2}-\epsilon)p_t-1$ when $p_t \leq \rho_B$, thus the expected utility of the lender is smaller in this case.
As such, if we fix the borrower strategy to be not liquidate, given that the expected payoff of the lender  when liquidating is either same or smaller than the expected payoff of not liquidating, we see that the strategy of not liquidating maximises the expected payoff of the lender in this case.

Thus, when the terminal exchange rate $r_t > \rho_B$, $\sigma \restriction_{E_B} =(\text{lend}, \text{correct contract},\\ (\text{not liquidate},\text{not liquidate}), x=1\text{M}, \text{open})$ is a subgame perfect equilibrium in $\gammano$.
When $r_t \leq \rho_B$, $\sigma \restriction_{E_B} =(\text{lend}, \text{correct contract}, (\text{not liquidate},\text{not liquidate}), \\ 
x=1\text{M}, \text{not open}, \text{open})$ is a subgame perfect equilibrium in $\gammano$. 
\end{proof}

\begin{proof}(Proof of~\Cref{thm:main_oracleless})
Follows directly from~\Cref{lem:honest_noliquid},~\Cref{lem:subgame_el}, and~\Cref{lem:subgame_eb}. 
\end{proof}

We end this section by commenting on the stable strategies we can achieve with this protocol.
As we can see from~\Cref{thm:main_oracleless}, one stable strategy profile is for a non-myopic borrower to always not liquidate and wait for the contract to be opened when the price of BTC rises above a certain threshold.
The lender would not open the contract and wait for the borrower to open the contract to release the funds to both sides.
Although this strategy profile might not correspond to the ``honest strategy" as in the case of the previous protocols, we do not believe it is an issue as we implicitly assume in our analysis that the borrower is, to some extent, non-myopic and thus would favour not liquidating over liquidating when both the expected payoffs are the same. 
There could be, however, a different stable strategy profile for a borrower with different preferences and risk propensity level.
We are aware that our assumption that the borrower is non-myopic is fairly strong and not the most realistic.
Nevertheless, accurately modelling these preferences is extremely complicated and challenging, and thus we view our work as a first step and leave enlarging our model to account for these preferences as an interesting direction for future work.
\end{document}